\documentclass[5p,times]{elsarticle}
\makeatletter
\def\ps@pprintTitle{%
 \let\@oddhead\@empty
 \let\@evenhead\@empty
 \def\@oddfoot{}%
 \let\@evenfoot\@oddfoot}
\makeatother

\usepackage{amsmath}
\usepackage{amsthm}
\usepackage{amssymb}
\usepackage{color,graphicx}

\usepackage{pgf}
\usepackage{tikz}
\usetikzlibrary{arrows,automata}

\usepackage[all]{xy}

\newtheorem{thm}{Theorem}

\newtheorem{lem}[thm]{Lemma}
\newtheorem{prop}[thm]{Proposition}

\theoremstyle{definition}
\newtheorem{myDef}[thm]{Definition}

\newcommand{\textcite}[1]{\citeauthor*{#1}~\cite{#1}}

\title{Weighted automata are compact and actively learnable}

\author[1,2]{Artem Kaznatcheev}
\author[3]{Prakash Panangaden}

\address[1]{Department of Biology, University of Pennsylvania, Philadelphia, USA}
\address[2]{Department of Computer Science, University of Oxford, Oxford, UK}
\address[3]{School of Computer Science, McGill University, Montreal, Canada}

\begin{document}

\begin{abstract}
    We show that weighted automata over the field of two elements can be exponentially more compact than non-deterministic finite state automata. 
    To show this, we combine ideas from automata theory and communication complexity. 
    However, weighted automata are also efficiently learnable in Angluin's minimal adequate teacher model in a number of queries that is polynomial in the size of the minimal weighted automaton. 
    We include an algorithm for learning WAs over any field based on a linear algebraic generalization of the Angluin-Schapire algorithm.
    Together, this produces a surprising result: weighted automata over fields are structured enough that even though they can be very compact, they are still efficiently learnable.
\end{abstract}

\maketitle

\section{Introduction}

Weighted automata (WAs) are an enriched model of finite state machines and define a natural representation of free monoids.  
They have received a lot of interest in the learning community because they provide an interesting way to represent and analyze sequence data, such as music~\cite{MMW09} or text and speech processing~\cite{MPR08}. 
\textcite{M09} provides a nice survey of algorithms related to weighted automata.

In this paper, we aim to expand the theoretical results known about the representational power and learnability of WAs.  We show that WAs over $\mathbb{Z}_2$ (WA2s) -- which can be viewed as word recognizers for regular languages in a natural way -- can be exponentially more compact than non-deterministic finite state automata (NFAs) and yet learnable in Angluin's~\cite{A87} queries and counter-examples (or minimal adequate teacher) model in a number of membership queries and counter-example queries that is polynomial in the size of the minimal weighted automaton.

With Theorem~\ref{thm:sep}, we show that there exists a family of languages where the minimal WA2s are exponentially smaller than the smallest NFAs. 
Unfortunately, in Theorem~\ref{thm:smallNFA} we show that there also exists an exponential separation in the other direction. 
This shows that one can sometimes, but not always, get a significantly more compact representation by using WAs. 
However, the compactness result is still interesting because there are efficient algorithms for minimizing WAs~\cite{WAmin} whereas finding a minimal NFA is PSPACE-complete~\cite{MS72}.

This makes weighted automata compact yet -- unlike NFAs~\cite{AK95,DLT04} -- structured enough to be actively learnable in Angluin's minimal-adequate teacher model.
In Section~\ref{sec:algo}, we show how to extend the Angluin-Schapire algorithm~\cite{A87,S91} to weighted automata over any field. 
As such, we show that although WAs can be exponentially smaller than NFAs (and thus also DFAs), they still have a structure that we can exploit for efficient learning. 
Since weighted automata correspond more closely to popular models like POMDPs and probabilistic automata~\cite{CT04}, this might open new avenues for learning algorithms of those representations.

\section{Formal background}
\label{sec:genAut}

\subsection{Finite state automata}

\begin{myDef}
	Given a fixed alphabet $\Sigma$ and finite dimensional vector space $\mathbb{F}^n$, a weighted automaton $A$ over $\mathbb{F}$ of size $n$ is given by:
	
\begin{equation}
M = \langle \alpha, \omega \in \mathbb{F}^n, \{M^\sigma \in \mathbb{F}^{n \times n} | \sigma \in \Sigma \} \rangle
\end{equation}
where $\alpha$ is the initial state, $\omega$ is a final measurement (or final state), and for each $\sigma \in \Sigma$ we have a corresponding transition matrix $M^\sigma$.  The function recognized by this automaton is given by:

\begin{equation}
	f_M(\sigma_1...\sigma_m) = \alpha^T M^{\sigma_1} \ldots M^{\sigma_m} \omega
\end{equation}
\label{def:WAgen}
\end{myDef}

When dealing with automata, it is useful to adapt a general matrix representation of the function they recognize:

\begin{myDef}
	Given a function $f: \; \Sigma^* \rightarrow \mathbb{F}$  the \emph{Hankel matrix} $H_f: \Sigma^* \times \Sigma^* \rightarrow \mathbb{F}$ of $f$ is: $H_f(u,v) = f(uv)$.
\end{myDef}

We will also talk about the \emph{restricted Hankel matrix} $H_f|_n: \Sigma^n \times \Sigma^n \rightarrow \mathbb{F}$ of $f$ to strings of length $n$.

The Hankel matrix allows us to come to grips with weighted automata and their size:

\begin{prop}[\cite{CP71,F74}]
	$\mathrm{rank}_\mathbb{F}(H_f) \leq n$ if and only if there exists a weighted automaton $A$ over $\mathbb{F}$ of size $n$ such that $f_A = f$.
	\label{thm:WAsize}
\end{prop}

If we are going to study weighted automata over $\mathbb{Z}_2$  (WA2) and non-deterministic finite state automata (NFA) together then it is best to express them in a common framework. 
To do this, we will define a generic finite state automaton (Definition~\ref{def:FSA}) and then see how augmenting this model with different acceptance criteria can produce NFAs (Definition~\ref{def:NDFSA}) or WA2s (Definition~\ref{def:WFSA}), or restricting the kinds of transitions can produce deterministic finite-state automata (DFA; Definition~\ref{def:DFSA}).

\begin{myDef}
	A \emph{finite state automaton} (FSA) is a tuple 
	$A = \langle Q, \Sigma, \delta: Q \times \Sigma \rightarrow 2^Q, S \subseteq Q, F \subseteq Q \rangle$
	where $Q$ is a finite set of states, $\Sigma$ is a finite alphabet, $\delta$ is the transition function, $S$ is a set of starting states, and $F$ is a set of final states.
	The size $|A|$ of the automaton is the number of states $|Q|$. 
	\label{def:FSA}
\end{myDef}

\begin{myDef}
	The dynamics of an FSA $A$ are defined by looking at $\mathrm{paths}_A: \Sigma^* \rightarrow 2^{Q^*}$
	where for $p \in Q^*$, $w \in \Sigma^*$, $q,q' \in Q$, and $a \in \Sigma$ we have the recursive definition:
	\begin{itemize}
		\item $\mathrm{paths}_A(\epsilon) = S$; and 
		\item $pqq' \in \mathrm{paths}_A(wa)$ if $pq \in \mathrm{paths}_A(w)$ 
		and $q' \in \delta(q,a)$.
	\end{itemize}
	We say that a \emph{path is accepting} if it ends in $F$, or formally: 
	$\mathrm{apaths}_A(w) \subseteq \mathrm{paths}_A(w)$ where $pq \in \mathrm{apaths}_A(w)$ if $q \in F$.
	\label{def:paths}
\end{myDef}

It will also be useful to have the following two refinements of paths:

\begin{myDef}
	Given an FSA $A$ and a state $q \in Q$ we say that a word $w \in \mathrm{past}(q)$ if $pq \in \mathrm{paths}(w)$ for some $p \in Q^*$.
\end{myDef}

In other words, $\mathrm{past}(q)$ is the set of all words that lead to $q$. In a similar vein, we can define:

\begin{myDef}
	Given an FSA $A$ and a state $q \in Q$ we say that $w \in \mathrm{future}(q)$ if $\exists v \in \mathrm{past}(q) \; p,r \in Q^* \; \mathrm{s.t} \; pqr \in \mathrm{apaths}(vw)$.
\end{myDef}

In other words, $\mathrm{future}(q)$ is the set of all words that lead from a reachable state $q$ to a state in $F$. 

Together, Definitions~\ref{def:FSA} and \ref{def:paths} specify a generic finite state automaton and how it runs.
What remains is how the automaton produces its corresponding recognized language.
This requires giving an acceptance criterion or membership criterion for the corresponding language.
Given an FSA $A$, we can get the traditional language $L^\mathrm{NFA}(A)$ recognized by a non-deterministic finite automaton (NFA) as:

\begin{myDef}
    An FSA $A$ is said to \emph{NFA-recognize} a language $L^\mathrm{NFA}(A)$ if
    \begin{equation}
        w \in L^\mathrm{NFA}(A) \iff |\mathrm{aparths}_A(w)| \geq 1.
    \end{equation}
    \label{def:NDFSA}
\end{myDef}

For shorthand -- and consistency with traditional nomenclature -- we say that an NFA $A$ recognizes language $L_A$ to means that an FSA $A$ NFA-recognizes a language $L$.
The reason for the unwieldy term ``NFA-recognize'' is because we want to provide a similar definition for weighted automata over $\mathbb{Z}_2$:

\begin{myDef}
    An FSA $A$ is said to \emph{WA2-recognize} a language $L^\mathrm{WA2}(A)$ if
    \begin{equation}
        w \in L^\mathrm{WA2}(A) \iff |\mathrm{apaths}_A(w)| = 1 \mod 2.
    \end{equation}
    \label{def:WFSA}
\end{myDef}

This allows us to give a machine view of the matrix- and function-based Definition~\ref{def:WAgen} of weighted automata:

\begin{prop}
Given a weighted automaton $M$ over $\mathbb{Z}_2$ of size $n$ computing the function $f_M$:
\begin{equation}
M = \langle \alpha, \omega, \{M^\sigma \;|\; \sigma \in \Sigma \}
\end{equation}

Let:

\begin{align}
    Q & = \{1,...,n\} \\
    \delta(q,\sigma) & = \{r \;|\; M^\sigma_{q,r} = 1\} \\
    S & = \{r \;|\; \alpha_r = 1\} \\
    F & = \{r \;|\; \omega_r = 1\}
\end{align}
\noindent then the FSA $A = \langle Q, \Sigma, \delta, S, F\rangle$ WA2-recognized $L$ if
\begin{equation}
    w \in L \iff f(w) = 1
\end{equation}
\label{prop:WAFun2Lang}
\end{prop}

\noindent This transition between linear algebraic and machine views of weighted automata is standard~\cite{WAbook}, but we include a proof for convenience:

\begin{proof}
    Note that it doesn't matter when we switch to mod 2: the matrix multiplication in Definition~\ref{def:WAgen} can be done over $\mathbb{R}$ until we multiply by the final measurement vector. 
    Since the transition function $\delta(\cdot,\sigma)$ is given by the matrix $M^\sigma$, we can just use matrix multiplication. 
    Multiplying $\alpha$ by transition matrices against is the same thing as counting the number of paths from $S$.
    Multiplication by $\omega$ adds up the paths that lead to final states $F$ and so computes $|\mathrm{apaths}(w)|$.
    Finally taking the mod 2 that we deferred completes our computation.
\end{proof}

As with NFAs, Proposition~\ref{prop:WAFun2Lang} allows us to shorten the unwieldy language of ``FSA $A$ WA2-recognizes the language $L$'' by the shorter and more traditional ``WA2 $A$ recognizes the language $L$``.

Note also that by the same argument as Proposition~\ref{prop:WAFun2Lang}, we could view the NFA from Definition~\ref{def:NDFSA} as a weighted automaton that uses the boolean semiring (`or' for addition, and `and' for multiplication) instead of over a field as in Definition~\ref{def:WAgen}.
In other words, NFAs can also be thought of as weighted automata over the boolean semiring. 
This is why when we discuss weighted automata in this article, we focus only on WAs over fields (and do not consider the more general setting of WAs over rings).

Finally, let us make the familiar definition of deterministic finite state automata by putting restrictions on $\delta$ and $S$:

\begin{myDef}
	An FSA $A$ is a \emph{deterministic finite automaton} (DFA) if it respects the restriction of a single start state ($|S| = 1$) and deterministic transitions: 
	\begin{equation}
	    \forall q \in Q, \; a \in \Sigma \; |\delta(q,a)| = 1;
	    \label{eq:deterministic}
	\end{equation}
	
	The DFA $A$ is said to recognize a language $L_A$ if
	
	\begin{equation}
		w \in L_A \subseteq \Sigma^* \iff |\mathrm{apaths}(w)| = 1.
	\end{equation}
	\label{def:DFSA}
\end{myDef}

Note that the DFA restrictions of a single start state and determinism (Equation~\ref{eq:deterministic}) imply that given a DFA $A$, any word $w$ defines only one path (i.e., $\forall w \in \Sigma^* \;\; |\mathrm{paths}_A(w)| = 1$) and this path is either accepting or not. 
This means that a DFA is also an NFA, and WA2. 

\subsection{Tools from communication complexity}

It will be useful to observe a link between the Hankel matrix and a concept from communication complexity:

\begin{myDef}
	The \emph{1-monochromatic rectangle covering} of a function $f: \{0,1\}^n \times \{0,1\}^n \rightarrow \{0,1\}$ is the smallest number $\chi_1(f)$ of pairs of sets (called rectangles) $A_i,B_i \subseteq \{0,1\}^n$ 
	for $1 \leq i \leq \chi_1(f)$ such that: 
	\begin{enumerate}
	\item for every $(x,y) \in A_i \times B_i$ we have $f(x,y) = 1$ (i.e., $A_i \times B_i$ is 1-monochromatic), and
	\item for every $(x,y) \in f^{-1}(1)$ we have at least one index $i \in \{1, ..., \chi_1(f)\}$ such that $(x,y) \in A_i \times B_i$.
	\end{enumerate}
	\label{def:rect}
\end{myDef}

Based on formalizing the argument in Hromkovi{\v{c}} and Schnitger~\cite{HS08} that views NFAs as a non-deterministic one-way communication protocol where the message sent by the first computer to the second corresponds to the state of the NFA, 
we can show that the 1-monochromatic rectangle covering (which is a kind of non-deterministic one-way communication protocol) lower bounds the size of NFAs:

\begin{prop}
	$|\mathrm{NFA}(f)| \geq \chi_1(H_f|_n)$ for any $n \in \mathbb{N}$.
	\label{prop:NFAchiLow}
\end{prop}

\begin{proof}
	Let $A$ be a minimal $\mathrm{NFA}$ recognizing $f$.
	For each state $q \in Q$ define $A_q = \mathrm{past}(q)$ and $B_q = \mathrm{future}(q)$, by the definition of $\mathrm{future}(q)$ for any $u \in A_q$ and $v \in B_q$ we have $f(uv) = 1$. 
	Therefore, the $\{A_q,B_q\}_{q \in Q}$ are 1-monochromatic rectangles. 
	Now, consider any $uv \in f^{-1}(1)$, say that $q \in Q_u$ if $\exists p \in Q^*$ such that $pq \in \mathrm{paths}(u)$. 
	Since $f(uv) = 1$, there must be at least one $q \in Q_u$ such that $v \in \mathrm{future}(q) = B_q$. 
	Therefore, the $\{A_q,B_q\}_{q \in Q}$ are a cover of the whole Hankel matrix, and hence any restricted submatrix is also covered.
\end{proof}

Another useful tool for proving lower bounds in communication complexity is:

\begin{myDef}
	The \emph{discrepancy} of a function $f: \{0,1\}^n \times \{0,1\}^n \rightarrow \{0,1\}$ is:
	
	\begin{equation}
		\mathrm{disc}(f) = \max_{A,B \subseteq \{0,1\}^n} \frac{1}{2^{2n}}\; \Bigg|\sum_{x \in A, y \in B} (-1)^{f(x,y)}\;\Bigg|
	\end{equation}
	\label{def:disc}
\end{myDef}

Definitions~\ref{def:rect} and~\ref{def:disc} relate nicely to each other by an extension of Lemma~13.13 from~\textcite{AB09}:

\begin{lem}
	$\chi_1(f) \geq \frac{|f^{-1}(1)|}{2^{2n}\mathrm{disc}(f)}$
	\label{lem:disc}
\end{lem}

\begin{proof}
	Since all the ones in our function can be covered by $\chi_1(f)$ squares, and a total of $|f^{-1}(1)|$ ones need to be covered, there must be at least one monochromatic rectangle $A \times B$ that covers the average number of ones or more. 
	This means that $|A||B| \geq |f^{-1}(1)|/\chi_1(f)$. 
	Now, since the discrepancy is a max over rectangles, we can pick $A \times B$ to lower bound it:

	\begin{align}
		\mathrm{disc}(f) & \geq \frac{1}{2^{2n}} |\sum_{x \in A, y \in B} (-1)^{f(x,y)}| \\
		& \geq \frac{1}{2^{2n}} |\sum_{x \in A, y \in B} -1| \\
		& \geq \frac{|A||B|}{2^{2n}} \\
		& \geq \frac{|f^{-1}(1)|}{2^{2n} \chi_1(f)}
	\end{align}
	
	\noindent where the second line follows from the first because the rectangle is 1-monochromatic. The last line can be rearranged to complete the proof.
\end{proof}

\section{Size of NFAs and WA2s}

We are interested in the following question: given a regular language $L$, what is the size of the smallest automaton $A$ with $L_A = L$? 
In particular, we will define $|\mathrm{NFA}(L)|$ to be the largest integer such that for any $\mathrm{NFA}$ $A$, if $L_A = L$ then $|A| \geq |\mathrm{NFA}(L)|$ and similarly for $|\mathrm{DFA}(L)|$, 
and $|\mathrm{WA2}(L)|$. 

\subsection{WA2s can be exponentially smaller than NFAs}

The gap between $|\mathrm{NFA}(L)|$  and $|\mathrm{WA2}(L)|$ can be exponentially large. Technically, this means that:

\begin{thm}
	There exists a family of regular languages $\{L_n\}$ such that 
	$|\mathrm{NFA}(L_n)| \in 2^{\Omega({|\mathrm{WA2}(L_n)|})}$
	\label{thm:sep}
\end{thm}

To find our separating family of languages, we will look at the inner-product function:

\begin{myDef}
	The \emph{n-bit inner product} is a function $\wedge^{\otimes n} : \{0,1\}^n \times \{0,1\}^n \rightarrow \{0,1\}$
	acting on two bit strings $x = x_1..x_n \in \{0,1\}^n$ and $y = y_1...y_n \in \{0,1\}^n$ as
	$x \wedge^{\otimes n} y = \sum_{i =1}^n x_1\cdot y_1 \mod 2$
\end{myDef}

Sometimes, when the size of $x$ and $y$ is obvious, we will omit the $\otimes n$.
Note that the number of zeros and ones in $\wedge$ is well balanced.

\begin{prop}
	$|(\wedge^{\otimes n})^{-1}(1)| = 2^{n - 1}(2^n - 1)$
	\label{prop:size}
\end{prop}

\begin{proof}
	Let $D = \{(x,y) | \exists i \in [n] \; \mathrm{s.t}\; x_i = y_i \}$ be the set of pairs of
	strings that overlap in at least one place. Now, consider a function $h$ defined on $D$ that given
	$(x,y)$ take the smallest index of overlap $i$ (i.e. for all $j < i, x_j \neq y_j$) and sends
	$x_i \rightarrow \bar{x_i}$ and $y_i \rightarrow \bar{y_i}$ this function is a bijection on $D$.
	However, note that if $\wedge^{\otimes n}(x,y) =  b$ then $\wedge^{\otimes n} h(x,y) = \bar{b}$.
	Thus, $\wedge$ has the same number of zeros and ones in $D$.
	
	The only pairs missing from $D$ are the ones of the form $(x,\bar{x})$ and there are $2^n$ such strings, so $|D| = 2^{2n} - 2^n$. 
	Finally, note that $x \wedge^{\otimes n} \bar{x} = 0$ thus $|(\wedge^{\otimes n})^{-1}(1)| = |D|/2$.
\end{proof}

\begin{lem}
	$\chi_1(\wedge^{\otimes n}) \geq 2^{n/2 -2}$
	\label{lem:wedge}
\end{lem}

\begin{proof}
	Example 13.16 in~\cite{AB09} shows that $\mathrm{disc}(\wedge^{\otimes n}) \leq 2^{-n/2}$ which combined with Lemma~\ref{lem:disc} and Proposition~\ref{prop:size} gives us $\chi_1(\wedge^{\otimes n}) \geq \frac{2^{n - 1}(2^n - 1)2^{n/2}}{2^{2n}} \geq 2^{n/2 - 2}$. 
\end{proof}

The inner-product allows us to define a special class of language families with an important property:

\begin{myDef}
	A language family $\{L_n\}$ is called an \emph{inner-product kernel family} if:
	\begin{equation}
		\forall n \; \forall x,y \in \{0,1\}^n \quad L_n(xy) = x \wedge^{\otimes n} y
	\end{equation}
\end{myDef}

Note that the above definition places no restriction on how $L_n$ behaves on words of length other than $2n$, so there are many inner-product kernel families based on the many ways languages can behave outside the kernels.

\begin{prop}
	If $\{L_n\}$ is an inner-product kernel family then $|\mathrm{NFA}(L_n)| \geq 2^{n/2 - 2}$
	\label{lem:NFA}
\end{prop}

\begin{proof}
	We use the communication complexity techniques from Proposition~\ref{prop:NFAchiLow}. 
	We can use any finite submatrix of $H_L$ to lowerbound $|\mathrm{NFA}(L)|$. 
	In particular, if for $L_n$ we look at the submatrix of $H_L$ with rows and columns indexed by strings of length $n$ then this submatrix is the same as the matrix for $\wedge^{\otimes n}$. 
	Thus, $|\mathrm{NFA}(L_n)| \geq \chi_1(\wedge^{\otimes n}) \geq 2^{n/2 - 2}$ where the first inequality is an application of the Proposition~\ref{prop:NFAchiLow} lowerbound technique and the second inequality is from Lemma~\ref{lem:wedge}.
\end{proof}

\begin{figure}
	\center
\begin{tikzpicture}[->,>=stealth',shorten >=1pt,auto,node distance=1.5cm,    semithick]
  \tikzstyle{every state}=[fill=white,draw=black,text=black]

  \node[initial,state](S){$s$};
  \node[state](M1)[below of=S]{$m_1$};
  \node[state](M2)[right of=M1]{$m_2$};
  \node (Mdot)[right of=M2]{$\cdots$};
  \node[state](Mn)[right of=Mdot]{$m_n$};
  \node[state,accepting](F)[above of=Mn]{$f$};

  \path (S) edge              node {1} (M1)
            edge [loop right] node {0,1} (S)
        (M1) edge node {0,1} (M2)
        (M2) edge node {0,1} (Mdot)
        (Mdot) edge node {0,1} (Mn)
        (Mn) edge node {1} (F)
        (F) edge [loop left] node {0,1} (F)
            ;
\end{tikzpicture}
	\caption{A picture of the weighted automaton used to prove Theorem~\ref{thm:sep}.}
	\label{fig:WAprod}
\end{figure}
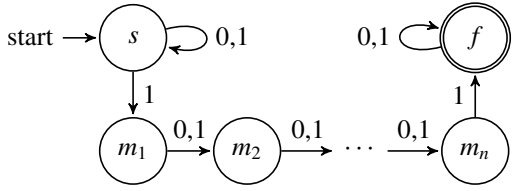

We finish the proof of Theorem~\ref{thm:sep} by noticing that the family of weighted automata in Figure~\ref{fig:WAprod} recognize languages in an inner-product kernel family but only have $n + 2$ states.
More formally:

\begin{prop}
Let $\text{WA}^\text{prod}_n$ be the weighted automaton in Figure~\ref{fig:WAprod}.
Given any $x,y \in \{0,1\}^n$: 
\begin{equation}
xy \in L_{\text{WA}^\text{prod}_n} \iff \sum_{i = 1}^n x_iy_i \mod 2 = 1.
\label{eq:propWAprod}
\end{equation}
\label{prop:WAprod}
\end{prop}

\begin{proof}
Any accepting path in $\text{WA}^\text{prod}_n$ must have the form $p \in s^*m_1m_2...m_nf^*$.
A path $p$ is caused by transitions corresponding to a word of the pattern: 
\begin{equation}
\{0,1\}^*1\{0,1\}^{n-1}1\{0,1\}^*
\label{eq:pattern}
\end{equation}

\noindent i.e., by a word that has two $1$s that are exactly $n$ letters apart.

Now, let us count the number of accepting paths for any $xy$.
The word $xy$ matches the pattern in Equation~\ref{eq:pattern} for each $1 \leq i \leq n$ such that $x_i = y_i = 1$ and for no other: i.e., only for the partition $\{0,1\}^{i-1}1\{0,1\}^{n-1}1\{0,1\}^{n-i}$.
Each of these partitions of $xy$ corresponds to a unique path, so the total number of accepting paths is $\sum_{i = 1}^n x_iy_i$ and
Equation~\ref{eq:propWAprod} follows from the acceptance criteria of WAs in Definition~\ref{def:WFSA}.
\end{proof}

\subsection{NFAs can be exponentially smaller than WA2s}
Unfortunately, there are also cases where the opposite happens and we do not have a small WA2 while a small NFA exists:

\begin{thm}
	There exists a family of regular languages $\{L_n\}$ such that 
	$|\mathrm{WA2}(L_n)| \in 2^{\Omega({|\mathrm{NFA}(L_n)|})}$
	\label{thm:smallNFA}
\end{thm}

\begin{proof}
	For this, consider a language family where for $u,v \in \{0,1\}^n$ $uv \in L_n$ if and only if $u \neq v$. If we look at the Hankel matrix of $L_n$ restricted to columns and rows of length $n$ then it is a matrix of all ones except with zeros on the diagonal. Clearly, this matrix has full rank, so by Theorem~\ref{thm:WAsize} $|WA2(L_n)| \geq 2^n$. 
	
	On the other hand, an NFA of size $2(n + 1)$ is given that recognizes a language consistent with $L_n$ in Figure~\ref{fig:smallNFA}.
	Notice that any accepting path in this NFA can only have been caused by a word of the pattern $\{0,1\}^*0\{0,1\}^{n-1}1\{0,1\}^*$ (left branch) or $\{0,1\}^*1\{0,1\}^{n-1}0\{0,1\}^*$ (right branch).
	When we restrict this to words $xy$ with $x,y \in \{0,1\}^n$, we see that one of the patterns is realized only if there is some $1 \leq i \leq n$ such that $x_i \neq y_i$.
\end{proof}

\begin{figure}
	\center

\begin{tikzpicture}[->,>=stealth',shorten >=1pt,auto,node distance=1.5cm, semithick]
  \tikzstyle{every state}=[fill=white,draw=black,text=black]

  \node[initial above,state](S){$s$};
  \node[state](L1)[left of=S]{$l_1$};
  \node[state](L2)[below of=L1]{$l_2$};
  \node (Ldot)[below of=L2]{$\vdots$};
  \node[state](Ln)[below of=Ldot]{$l_n$};
  \node[state](R1)[right of=S]{$r_1$};
  \node[state](R2)[below of=R1]{$r_2$};
  \node (Rdot)[below of=R2]{$\vdots$};
  \node[state](Rn)[below of=Rdot]{$r_n$};
  \node[state,accepting](F)[right of=Ln]{$f$};

  \path (S) edge              node {0} (L1)
  			edge	node {1} (R1)
            edge [loop below] node {0,1} (S)
        (L1) edge node {0,1} (L2)
        (L2) edge node {0,1} (Ldot)
        (Ldot) edge node {0,1} (Ln)
        (Ln) edge node {1} (F)
        (R1) edge node {0,1} (R2)
        (R2) edge node {0,1} (Rdot)
        (Rdot) edge node {0,1} (Rn)
        (Rn) edge node {0} (F)
        (F) edge [loop above] node {0,1} (F)
            ;
\end{tikzpicture}
	
	\caption{A picture of the NFA used in the proof of Theorem~\ref{thm:smallNFA}}
	\label{fig:smallNFA}
\end{figure}
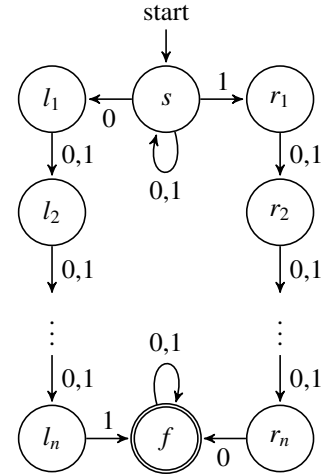

\section{Efficient active learning algorithm for weighted automata}
\label{sec:algo}

Deterministic finite state automata (DFAs) are not passive learnable: i.e., DFAs are known to be difficult to PAC-learn from randomly drawn labeled examples in any representation~\cite{KV94}. 
However, we can instead consider a model with active learning that instead of random labeled examples has the following two types of queries:
\begin{enumerate}
\item for any string $x \in \Sigma^*$ we can do a membership query to get $f(x)$. This is the active learning component, since the algorithm generates the query to ask, and 
\item given a candidate weighted automaton $A$, we can ask if it is correct with a counter-example query. If $A$ computes $f$ (i.e. $f_A = f$) then the teacher will say ``CORRECT", otherwise the teacher will return a counter-example $z$ such that $f_A(z) \neq f(z)$. If a teacher is unavailable then this can alternatively be replaced by random sampling if we want a PAC-like model, and would correspond to the non-active part of learning. 
\end{enumerate}
This is Angluin's queries and counter-examples or `minimal adequate teacher' (MAT) model~\cite{A87}.
\textcite{A87} famously showed that -- in the MAT model -- regular languages are efficiently learnable in the size of their minimal DFA representation. 
Later, \textcite{S91} improved the efficiency of Angluin's algorithm for learning DFAs.
In this section, we show how to adapt the Angluin-Schapire algorithm from learning DFAs to learning WAs over any field $\mathbb{F}$. 

For the rest of the section, suppose we are trying to learn an unknown function $f: \Sigma^* \rightarrow \mathbb{F}$ with Hankel matrix $H: \Sigma^* \times \Sigma^* \rightarrow \mathbb{F}$. 

\subsection{Initialization}

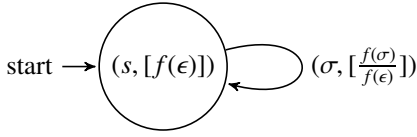
\begin{figure}
\centering
\begin{tikzpicture}[->,>=stealth',shorten >=1pt,auto,node distance=1.5cm, semithick]
  \tikzstyle{every state}=[fill=white,draw=black,text=black]

  \node[initial,state](S){$(s,[f(\epsilon)])$};

  \path (S) edge [loop right] node {$(\sigma, [\frac{f(\sigma)}{f(\epsilon)}])$} (S);
\end{tikzpicture}
	\caption{Initial weighted automaton. There is a single state that is initial and outputs its weight times $f(\epsilon)$. There is a self-loop for each letter $\sigma \in \Sigma^*$ weighted by $\frac{f(\sigma)}{f(\epsilon)}$. Note that we are using the WLOG assumption that $f(\epsilon) \neq 0$.}
	\label{fig:initial}
\end{figure}

At all times, our algorithm will keep track of two finite sets $S,E \subseteq \Sigma^*$ of equal size ($|S| = |E|$). 
$S$ will be prefix closed and we will call its elements states. 

For convenience, we will define a function $F: S \rightarrow \mathbb{F}^E$. 
If we view $F$ as a matrix, then it is a restriction of $H$ to $S$ and $E$, i.e. $F = H(S,E)$ or more explicitly for $s \in S$ and $e \in E$, $F(s,e) = f(se)$. 
Our algorithm will ensure that $F$ is full rank, i.e. $\mathrm{rank}_{\mathbb{F}}(F) = |S|$.

We will start with $S = E = \{\epsilon\}$ and without loss of generality assume that $f(\epsilon) \neq 0$ (if it is equal zero then just replace $f$ by $f + 1$, learn that, and then subtract $1$ from each value in the final/measurement state). 
See Figure~\ref{fig:initial} for the initial automaton.
This initialization requires one membership query to learn $f(\epsilon)$.

\subsection{Automaton corresponding to matrix $F$}

For each $\sigma \in \Sigma$, consider $F^\sigma: S \rightarrow \mathbb{F}^E$ where $F^\sigma(s,e) = f(s\sigma e)$. 
Since $F$ has full rank, we know that its columns form a basis for $\mathbb{F}^E$.
Thus, every other vector $F^\sigma(s) \in \mathbb{F}^E$ can be expressed as some linear combination of the $F(s')$ for $s' \in S$.
Define $T^\sigma: S \times S \rightarrow \mathbb{F}$ as the matrix that stores the coefficients of these linear combinations: i.e., define $T^\sigma$ such that for every $s \in S$ we have $F^\sigma(s) = \sum_{s' \in S} T^\sigma_{s,s'} F(s')$.

This allows us to define the corresponding weighted automaton over $\mathbb{F}$ (see Definition~\ref{def:WAgen}) on state space $\mathbb{F}^S$. 
Let the weighted automaton $T$ have...
\begin{itemize}
\item initial state $\alpha$ such that $\alpha(\epsilon) = 1$ and $\alpha(s) = 0$ if $s \neq \epsilon$, 
\item final/measurement state $\omega = F(\cdot,\epsilon)$ (i.e., the row of $F$ corresponding to $\epsilon \in E$, and 
\item transition matrices $T^\sigma$.
\end{itemize}

\subsection{Learning from counter-example query}

Now, suppose we tried this automaton $T$ and our teacher returned a counter-example $z$. 
We will use this counter-example to find strings to extend $S$ and $E$ and thus increase the rank of our matrix $F$. 
Now for each $1 \leq i \leq |z| + 1$ consider the partitions $z = z_{< i}\sigma_i z_{> i}$. 
For each $z_{< i}$ define $Z_i: S \rightarrow \mathbb{F}$ to be the state of our candidate automaton when we run it on $z_{<i}$:

\begin{equation}
Z_i = \alpha^T T^{z_1} T^{z_2} \ldots T^{z_{i - 1}}.
\end{equation}

Let $f_i = \sum_{s \in S} Z_i(s)f(s\sigma_i z_{> i})$. 
From our definition, we know that $f_1 = f(z) \neq f_T(z) = f_{|z| + 1}$, thus as we increase $i$ there must be some point $k$ where $f_k \neq f_{k + 1}$. Find this point by using binary search on $i$. 
This requires at most $|S|\lceil \log (|z|) \rceil$ membership queries to $f$.

Let us write out $f_{k + 1}$:

\begin{eqnarray}
	f_{k + 1} &=& \sum_{s' \in S}Z_{k + 1}(s') f(s'z_{> k}) \\ &=& \sum_{s,s' \in S} T^{\sigma_k}_{s,s'} Z_k(s) f(s'z_{> k}) \label{eq:longkplusone}
\end{eqnarray}

Now, proceed by contradiction: if $\forall s \in S$ we have
$f(s\sigma_kz_{> k}) = \sum_{s' \in S} T^{\sigma_k}(s,s')f(s'z_{> k})$ then
\begin{eqnarray}
f_k &=& \sum_{s \in S} Z_k(s) f(s\sigma_k z_{> k})\\ &=& \sum_{s \in S} Z_k(s) \sum_{s' \in S} T^{\sigma_k}_{s,s'}f(s'y_k) = f_{k + 1}
\end{eqnarray}
where the last equality follows from Equation~\ref{eq:longkplusone} and contradicts $f_k \neq f_{k + 1}$. Thus, there must be some $s^* \in S$ such that
$f(s^*\sigma_kz_{> k}) \neq  \sum_{s' \in S} T^{\sigma_k}_{s^*,s'}f(s'z_{> k})$. 

Now, consider an $s\sigma \in S$ then
\begin{equation}
F(s\sigma) = F_{\sigma}(s) =
\sum_{s'}T^{\sigma}_{s,s'}F(s')
\end{equation}
\noindent but since the $F(s)$ are linearly independent, we must have that $T^\sigma_{s,s\sigma} = 1$ and for $s' \neq s\sigma$ we must have $T^\sigma_{s,s'} = 0$. 
Plugging this into our contradiction assumption, we see that for $s\sigma_k \in S$ we have $\sum_{s' \in S} T^{\sigma_k}_{s,s'}f(s'z_{> k}) = f(s\sigma_k z_{> k})$. 
Therefore, our $s^*\sigma_k \not\in S$. 
Now, we can add $s^*\sigma_k$ to $S$ and $z_{> k}$ to $E$ to get a new linearly independent row and column and increase the rank of our matrix by 1.

\subsection{Termination}

Since our candidate automaton agrees with $f$ on every value in $F$, it must be that the real weighted automaton corresponding to $f$ must have more states than $\mathrm{rank}(F)$. 
At every counter-example query, we increase our rank by one, so if our world $f$ is represented by a minimum weighted automaton with $n$ states then after $n - 1$ counter-example queries we must have $rank(F) = rank(H_f)$.
Since our automaton agrees with $f$ on every value in $F$, the $n$th counter-example query gets it ``CORRECT".
If $m$ is the length of the longest-counterexample then the total number of membership queries is less than $1 + {n \choose  2} \lceil \log m \rceil$.

\section{Discussion and Conclusion}

As far as we know, this is the first time it has been show that weighted automata (WAs) can be exponentially smaller than NFAs. 
Together with the learning algorithm, this produces a somewhat surprising result: weighted automata are structured enough that even though they are compact, they are still efficiently learnable.
This also means that some languages where the minimal DFAs and NFAs are exponentially bigger than the minimal WAs can be learned much faster using the WA representation. 

This is not the case for NFAs.
Although several algorithms have been developed for learning NFAs in the minimum adequate teacher model~\cite{Y94,BHKL09}, the results for NFAs differ from the case of WAs in two fundamental ways:
\begin{enumerate}
    \item The algorithms for learning NFAs are not guaranteed to return a minimal NFA that recognizes the language.
    In fact, they return a special kind of NFAs called residual finite state automata (RFSAs)~\cite{DLT04,BHKL09}. 
    These RFSAs are always the same size or larger than NFAs and in some cases are exponentially larger than the minimal NFA that recognizes a language~\cite{DLT04}.
    \item The number of queries required for learning these RFSAs is not polynomial in the size of the minimal NFA nor the minimal RFSA, but only polynomial in the size of the minimal DFA.
    So although RFSAs can be exponentially more compact that DFAs, this does not necessarily provide a speed-up for learning those RFSAs.
    In fact, there are hardness results suggesting that one cannot learn NFAs or RFSAs in a number of queries that is polynomial in the size of smallest NFA recognized the language~\cite{AK95} nor polynomial in the size of smallest RFSA recognizing the language~\cite{DLT04}.
\end{enumerate}

\noindent In contrast, we show that a minimal WA can be learned in a number of queries that is polynomial in the size of the minimal WA corresponding to that unknown function.
Since WAs are always smaller than DFAs and sometimes exponentially smaller, that means that learning WA2s replaces the standard Angluin-Schapire algorithm~\cite{A87,S91} for learning regular languages. 
In the cases where WAs are the same size as DFAs, we can achieve the same performance, and in the cases in which WAs are more compact, we provide exponential savings in terms of queries used.

\section*{Acknowledgements}

We are indebted to helpful discussion with Borja Balle and Doina Precup.
The paper also benefited from the feedback of several anonymous reviewers.
The work began when A. Kaznatcheev was at the School of Computer Science, McGill University 
and completed thanks to the generous support of a James S. McDonnell Foundation Postdoctoral Fellowship for Understanding Dynamic \& Multi-scale Systems.
P. Panangaden was supported by NSERC (Canada).

\section*{References}

\bibliographystyle{plainnat}
\bibliography{WAref}

\end{document}